\documentclass[copyright,creativecommons]{eptcs}
\providecommand{\event}{Linearity-TLLA 2018} 
\usepackage{underscore}           
\usepackage[T5,T1]{fontenc}
\usepackage[utf8]{inputenc}

\usepackage{amsmath}
\usepackage{amsthm}
\usepackage{amssymb}
\usepackage{bussproofs}
\usepackage{csquotes}
\usepackage{cmll}
\usepackage{xspace}
\usepackage{enumitem}
\usepackage{tikz}
\usepackage{subcaption}

\usepackage{color}

\definecolor{darklavender}{rgb}{0.30, 0.21, 0.39}
\definecolor{amber}{rgb}{1.0, 0.75, 0.0}
\definecolor{lavenderindigo}{rgb}{0.58, 0.34, 0.92}
\tikzstyle{vertex}=[circle,fill=black,minimum size=7pt,inner sep=0pt]
\tikzstyle{bigvertex}=[circle,draw,thick,fill=black!5,minimum size=16pt,inner sep=0pt]

\title{Coherent Interaction Graphs}
\author{{{\fontencoding{T5}\selectfont Lê Thành Dũng \textsc{Nguyễn}}\thanks{Part
    of this work was carried out when the first author was a student at École
    normale supérieure de Paris.}}\;\thanks{Partially
  supported by the ANR project ELICA
  (ANR-14-CE25-0005).}\qquad\qquad\qquad Thomas \textsc{Seiller}\footnotemark[2]
  \institute{Laboratoire d'informatique de Paris Nord (LIPN) \qquad\quad CNRS, UMR 7030 -- LIPN\\
  Université Paris 13, Sorbonne Paris Cité, Villetaneuse, France}
\email{nltd@nguyentito.eu \qquad\qquad\qquad\qquad\qquad\quad seiller@lipn.fr}
}
\def\titlerunning{Coherent Interaction Graphs}
\def\authorrunning{L.\ T.\ D.\ {\fontencoding{T5}\selectfont Nguyễn} \& T.\ Seiller}
\begin{document}
\maketitle

\begin{abstract}
  We introduce the notion of coherent graphs, and show how those can be used to
  define dynamic semantics for Multiplicative Linear Logic (MLL) extended with
  non-determinism. Thanks to the use of a coherence relation rather than mere
  formal sums of non-deterministic possibilities, our model enjoys some
  finiteness and sparsity properties.

  We also show how studying the semantic
  types generated by a single \enquote{test} within our model of MLL naturally
  gives rise to a notion of proof net, which turns out to be exactly Retoré's
  R\&B-cographs. This revisits the old idea that multplicative proof net
  correctness is interactive, with a twist: types are characterized not by a set
  of counter-proofs but by a single non-deterministic counter-proof.
\end{abstract}

\theoremstyle{plain}
\newtheorem{theorem}{Theorem}
\newtheorem{proposition}[theorem]{Proposition}
\newtheorem{corollary}[theorem]{Corollary}
\newtheorem{lemma}[theorem]{Lemma}
\newtheorem{conjecture}{Conjecture}

\theoremstyle{definition}
\newtheorem{definition}[theorem]{Definition}

\theoremstyle{remark}
\newtheorem{remark}[theorem]{Remark}
\newtheorem{example}[theorem]{Example}

\newcommand{\plugging}{\mathbin{\square}} 
\newcommand{\symmdiff}{\mathbin{\triangle}} 
\newcommand{\execution}{\mathbin{::}} 
\newcommand{\Cyc}{\mathrm{Cyc}}
\newcommand{\cohspace}[1]{\mathrm{Coh}(#1)}
\newcommand{\MLL}{{\sc mll}\xspace}
\newcommand{\project}[1]{\mathfrak{#1}}
\newcommand{\graph}[1]{\mathrm{G}(\project{#1})}
\newcommand{\cycles}[1]{\mathrm{C}(\project{#1})}
\newcommand{\coherence}[1]{\coh_{\project{#1}}}
\newcommand{\coherencespace}[1]{\cohspace{\project{#1}}}
\newcommand{\CycProject}[1]{\Cyc(\project{#1})}
\newcommand{\cond}[1]{\mathbf{#1}}

\newcommand{\MEcoh}[1][\Omega]{\mathrm{cohM}_{\mathrm{fin}}(#1)}
\newcommand{\mfin}[1]{\mathrm{M}_{\mathrm{fin}}(#1)}
\newcommand{\incohplus}{\mathbin{\boxplus}}

\section{Introduction}

Dynamic semantics of proofs originate in Girard's \emph{Geometry of
  Interaction \emph{(GoI)} program}~\cite{Girard-towards}, whose aim was to
provide a semantic account for the process of cut-elimination. Indeed, while
the proofs-as-programs correspondence expresses that $\beta$-reduction of
lambda-terms corresponds to cut-elimination, its extension to categorical models
(the so-called Curry--Howard--Lambek correspondence) arguably fails to fully
reflect these dynamics, as it represents those operations as a simple equality.

Since then, GoI has been developed into many directions to give various accounts
of cut-elimination without syntactic rewriting. One aspect which has received
less attention is the construction of models of linear logic where morphisms
are generalized programs on which a GoI-style execution procedure can be
defined, and serves to compose morphisms. Girard's successive papers on GoI all
investigate such models, using operator algebras as the spaces of generalized
programs; recently, Seiller has simplified Girard's models using graphings (a
measure-theoretic extension of graphs) instead of operators~\cite{Seiller2017}.
Note that discrete graphs suffice to obtain a model of linear logic without
exponentials or quantifiers~\cite{Seiller2012,Seiller2016}.

The aforementioned models all start from an untyped universe of programs, and
interpret formulae as specifications for the behavior of programs: they can be
seen as realizability models with operators/graphs as realizers. Futhermore,
these specifications will be given as batteries of \enquote{tests} or
\enquote{counter-proofs} which are themselves generalized programs. 

This idea comes from the theory of multiplicative proof nets, where correct
proofs have to be distinguished out of a set of \enquote{proof structures} by
means of a \emph{correctness criterion}. Girard observed that by representing a
proof structure as a permutation $\sigma$, he could express the correctness
criterion as the cyclicity of $\sigma\tau_i$ for all $i$, where $\{\tau_i \mid i
\in I\}$ is a set of permutations depending on the formula being proved. This
led to a prototypical GoI model~\cite{multiplicatives}, limited to
Multiplicative Linear Logic (MLL), with permutations as programs. This model
internalizes \enquote{the proof $\sigma$ passes the test $\tau$} as a symmetric
relation $\sigma \perp \tau$, defined as \enquote{$\sigma\tau \text{ cyclic}$}.
In general, in GoI models, semantic types are defined using this kind of
symmetric relation, called an \emph{orthogonality}.

\paragraph{Contents of the paper.}

In this paper, we revisit this connection between proof net correctness and GoI
orthogonality in the MLL fragment. One issue is that, for a given MLL formula, 
the number of corresponding tests is exponential in the size of the formula -- and 
by replacing the formula by a huge test of sets, we have forgotten that they all
come from the same succinct data. We argue that a formula should correspond to a
\emph{non-deterministic superposition} of tests, and propose a variant of
Seiller's interaction graphs which admits \emph{sparse} representations of
non-deterministic sums of programs; the main ingredient is a \emph{coherence
  relation} between edges, hence the title. Instead of an exponential family of
tests, we get a single non-deterministic test of polynomial size.

Remarkably, this development -- which does not depend on knowing a priori
anything about MLL correctness criteria -- leads us to recover a previously
known but rarely used version of proof nets, introduced by
Retoré~\cite{retore2003}, together with its correctness criterion. Let us note 
that Ehrhard suggested~\cite{ehrhard2014} exploring connections between 
this criterion and GoI. This result which provides a correspondence between
cographs and formulae, suggest further work providing a tighter correspondence
between proof and graph theory by exploiting the notion of \emph{modular 
decomposition} of graphs (see e.g.~\cite{modular}).

The coherence allows for the representation of sums of programs, and therefore
provides a GoI account of MLL+\textsc{Sum}. This suggests extensions of
the framework to e.g. Differential Linear Logic (with or without promotion) \cite{ehrhard2006,ehrhard2017} which requires the use of the \textsc{Sum}.
But this is not the only possible use of a coherence relation: we will also 
describe in this paper how one can, in coherent interaction graphs, express a 
restriction to \emph{simple} paths -- i.e.\ paths visiting each vertex at most once 
-- to ensure that our generalized programs remain finite objects when they are
composed together. This was in fact the motivation for coherent graphs when they
were originally introduced in Seiller's PhD thesis~\cite{seillerthese}.

\paragraph{Notations.}

We write $\mathcal{P}(A)$ for the powerset of $A$. We distinguish two notions of
disjoint union: $A \sqcup B$ is $A \cup B$ with the implied condition that $A
\cap B = \emptyset$, while $A + B$ is the coproduct in the category of sets (one
may define $A + B = \{0\} \times A \sqcup \{1\} \times B$).

By a graph $G$, we mean a directed multigraph, i.e.\ a set of vertices $V(G)$
and a set of edges $E(G)$ with source and target maps $s_G, t_G : E(G) \to
V(G)$; $s_G \times t_G$ need not be injective. We write $G \equiv G'$ when $V(G)
= V(G')$ and there is a bijection $\varphi : E(G) \to E(G')$ such that $s_G =
s_{G'} \circ \varphi$, $t_G = t_{G'} \circ \varphi$, that is, when there is a
graph isomorphism between $G$ and $G'$ acting as the identity on vertices. This
is the right notion of \enquote{equality} (not merely isomorphism!) of graphs
for our purpose. We denote $\emptyset_V$ the graph with no edges on the vertex 
set $V$.

A \emph{path} is a finite sequence $v_0, e_1, v_1, \ldots, e_n, v_n$ with $v_i
\in V(G)$, $e_j \in E(G)$, $s_G(e_i) = v_{i-1}$, $t_G(e_i) = v_i$. Cycles are defined 
as paths $v_0, e_1, v_1, \ldots, e_n, v_n$ with $v_n=v_0$. 

Two edges $e$ and $f$ 
are \emph{incident} if they have a common endpoint. This does not depend on the 
directions of $e$ and $f$.

\section{From execution of matchings to interaction graphs}
\label{sec:ig}

In this section, we recall the second author's previous work on interaction graphs 
\cite{Seiller2012,Seiller2016}, in a particular case (empty pole), and omitting proofs.

\subsection{Multiplicative Linear Logic proofs as matchings}

Let $F$ be a formula of (unit-free) MLL. Any cut-free proof of $F$ determines a
fixed-point-free involution -- the $\sigma$ from the introduction -- on the
atoms of $F$, exchanging two occurrences of atoms when they were introduced by
the same axiom rule. It is well-known that for most purposes, one may in fact
forget about the syntactic proof and work only with the involution -- in fact
this is one way of looking at the theory of MLL proof nets. In particular,
cut-elimination can be carried out directly on involutions \cite{Seiller-Axioms}.

To do so, it is more convenient to see these fixed-point-free involutions as
\emph{matchings}, by which we mean\footnote{Properly speaking, these are
  actually called \emph{1-regular graphs} in graph theory. The usual meaning of
  \enquote{matching} (resp.\ \enquote{perfect matching}) is a subset of edges of
  a given graph such that each vertex is incident to at most (resp.\ exactly)
  one matching edge -- that is, perfect matchings correspond to 1-regular
  \emph{subgraphs}.} undirected graphs whose vertices all have degree 1: the
vertices are the atoms, and the edges are given by the orbits of the involution.
Given two cut-free proofs $\pi : (\vdash \Gamma, F)$ and $\pi' : (\vdash
F^\perp, \Delta)$, inducing the respective matchings $M$ and $M'$, to find the
matching for $\mathbf{cut}(\pi,\pi')$, we first \emph{plug} $M$ and $M'$
together:

\begin{definition}
  Let $G, H$ be two graphs; their vertex sets $V(G)$ and $V(H)$ might not
  be disjoint. $G \plugging H$ then denotes the graph $(V(G) \cup V(H), E(G)
  + E(H))$, equipped with the partition $\{E(G), E(H)\}$ of its edge set,
  in order to remember the provenance of edges.

  A path in $G \plugging H$ is \emph{alternating} if, for any two consecutive
  edges, exactly one of them is in $E(G)$.
\end{definition}
\begin{remark}
  There is an obvious notion of ternary plugging of graphs, with a
  tripartition of its edge set, which will also turn out to be useful
  later. It will be denoted $F \plugging G \plugging H$ for the graphs $F,G,H$.
\end{remark}

Here, $V(M) \cap V(M') \neq \emptyset$ because we identify the vertices
corresponding to atoms in $F$ with their duals in $F^\perp$. Translating
cut-elimination on the level of matchings just consists of following the
\emph{alternating paths} in the resulting graph $M \plugging M'$: see the example
in Figure~\ref{fig:ig}. This operation leads to the definition of a graph whose 
edges are exactly the alternating paths between $M$ and $M'$; the latter will be 
called the \emph{execution} of $M$ and $M'$, and denoted as $M \execution M'$
(Definition \ref{def:exec}).

The alternation condition can be justified through an analogy with game
semantics. An involution $\sigma$ of a set $S$ can be thought of as a
\enquote{strategy} over the \enquote{arena} $S$: if my opponent plays $x$, I
answer with $\sigma(x)$ -- in other words, I follow the incident edge of the
matching and answer with the other endpoint. Then the condition reflects an
alternation of moves between two players; execution is \enquote{composition}
(taking $G \plugging H$ and following paths) and \enquote{hiding} (keeping only
the symmetric difference $V(G) \symmdiff V(H)$).

In general, plugging two matchings results in an undirected graph whose vertices
have degree 1 or 2, i.e.\ a disjoint union of alternating paths and cycles. When
restricting to interpretations of MLL proofs, \emph{cycles never appear}. (This
corresponds to the strong normalization property for MLL proof nets.)

\begin{figure}
  \centering
  \begin{subfigure}{9cm}
    \centering
    \begin{prooftree}
      \AxiomC{}
      \UnaryInfC{$\vdash {\color{red}A^\bot, A}$}
      \AxiomC{}
      \UnaryInfC{$\vdash {\color{blue}A^\bot, A}$}
      \BinaryInfC{$\vdash {\color{red}A^\bot}, {\color{blue}A^\bot},
        {\color{red}A} \otimes {\color{blue}A}$}
      \UnaryInfC{$\vdash {\color{red}A^\bot} \parr {\color{blue}A^\bot},
        {\color{red}A} \otimes {\color{blue}A}$}
    
      \AxiomC{}
      \UnaryInfC{$\vdash {\color{amber}A^\bot, A}$}
      \AxiomC{}
      \UnaryInfC{$\vdash {\color{darklavender}A^\bot, A}$}
      \BinaryInfC{$\vdash {\color{darklavender}A^\bot}, {\color{amber}A^\bot},
        {\color{amber}A} \otimes {\color{darklavender}A}$}
      \UnaryInfC{$\vdash {\color{darklavender}A^\bot} \parr {\color{amber}A^\bot},
        {\color{amber}A} \otimes {\color{darklavender}A}$}
      \RightLabel{(Cut)}
      \BinaryInfC{$\vdash {\color{red}A^\bot} \parr
        {\color{blue}A^\bot}, {\color{amber}A} \otimes {\color{darklavender}A}$}
    \end{prooftree}
      \begin{tikzpicture}
        \node[vertex,fill=red] (a) at (0,0.8) {};
        \node[vertex,fill=blue] (b) at (1,0.8) {};
        \node[vertex,fill=red] (c) at (2,0.8) {};
        \node[vertex,fill=blue] (d) at (3,0.8) {};
        \draw[thick] (a) to[bend left] (c);
        \draw[thick] (b) to[bend left] (d);

        \node[vertex,fill=darklavender] (c') at (2,0) {};
        \node[vertex,fill=amber] (d') at (3,0) {};
        \node[vertex,fill=amber] (e') at (4,0) {};
        \node[vertex,fill=darklavender] (f') at (5,0) {};
        \draw[thick] (c') to[bend right] (f');
        \draw[thick] (d') to[bend right] (e');

        \draw[dotted, thick] (c) -- (c');
        \draw[dotted, thick] (d) -- (d');
      \end{tikzpicture}
    \end{subfigure}~~~~~~~~~{\LARGE $\leadsto$}~\begin{subfigure}{5cm}
      \centering
          \begin{prooftree}
      \AxiomC{}
      \UnaryInfC{$\vdash {\color{red}A^\bot},{\color{darklavender}A}$}
      \AxiomC{}
      \UnaryInfC{$\vdash {\color{blue}A^\bot},{\color{amber}A}$}
      \BinaryInfC{$\vdash {\color{red}A^\bot}, {\color{blue}A^\bot},
        {\color{amber}A} \otimes {\color{darklavender}A}$}
      \UnaryInfC{$\vdash {\color{red}A^\bot} \parr
        {\color{blue}A^\bot}, {\color{amber}A} \otimes {\color{darklavender}A}$}
    \end{prooftree}
      \begin{tikzpicture}
        \node[vertex,fill=red] (a) at (0,0) {};
        \node[vertex,fill=blue] (b) at (1,0) {};

        \node[vertex,fill=amber] (e') at (2,0) {};
        \node[vertex,fill=darklavender] (f') at (3,0) {};

        \draw[thick] (a) to[bend left] (f');
        \draw[thick] (b) to[bend left] (e');
      \end{tikzpicture}
    \end{subfigure}
    \caption{Cut-elimination reflected on matchings}
    \label{fig:ig}
\end{figure}

\subsection{Interaction graphs in a nutshell}

Let us now generalize by considering arbitrary \emph{directed graphs}
instead of matchings.
\begin{definition}
  \label{def:exec}
  The \emph{execution} $G \execution H$ of two graphs $G, H$ 
  has as vertex set the symmetric difference $V(G) \symmdiff V(H)$ and contains,
  for each alternating path from $u$ to $v$ in $G \plugging H$, an edge from $u$
  to $v$ ($u, v \in V(G\execution H)$).
\end{definition}

Execution defines a well-behaved composition for graphs thanks to this
associativity property:

\begin{proposition}
  If $V(F) \cap V(G) \cap V(H) = \emptyset$, then $(F::G)::H \equiv F::(G::H)$.
\end{proposition}

\begin{remark}
  Note that this is not just an isomorphism: on both sides of the $\equiv$ sign,
  the vertex set is the same, namely the union of $V(F)$, $V(G)$ and $V(H)$
  minus all pairwise intersections.
\end{remark}

A special case of execution allows us to represent the tensor product of MLL as
an operation on interaction graphs. For two graphs $G$ and $H$, if $V(G) \cap
V(H) = \emptyset$, then $G :: H$ is in fact a sort of disjoint union, with all
alternating paths having length~1. This defines the tensor product on graphs. 
 (In general, to ensure $V(G) \cap V(H) = \emptyset$, we
will resort to replacing $G$ and $H$ by isomorphic copies.)

\begin{definition}
For two graphs $G$ and $H$ such that $V(G) \cap V(H) = \emptyset$, we
define $G \otimes H$ as the graph $G \execution H$.
\end{definition}

This would be enough to build a degenerate model of MLL -- a compact closed
category -- with vertex sets as objects and graphs as morphisms. But instead, we
want a model in which $\parr \neq \otimes$.

\subsection{Building a model with orthogonality}

This requires something akin to the correctness criterion which distinguishes
$\parr$ from $\otimes$ in proof nets. In our case, it will be through a notion
of \emph{orthogonality}, accounting for linear negation at the level of graphs
-- as mentioned in the introduction (see also~\cite{Seiller2012,
  Seiller-Axioms}), this is related to correctness criteria. But our definition
of orthogonality can also be motivated by \enquote{interactive} considerations:
we saw earlier that plugging two matchings coming from MLL proofs with the right
types never led to cycles.

\begin{definition}
  For two graphs $G$ and $H$ with $V(G) = V(H)$, we write $G \perp H$ when $G
  \plugging H$ contains no alternating cycle. If $X$ is a set of graphs on a
  common vertex set, we define $X^\perp = \{ G \mid \forall H \in X,\, G \perp H
  \}$.
\end{definition}

In our model, the objects interpreting MLL types will be sets of graphs over a
common vertex set which can be written as $X^\perp$ for some $X$; they are
called \emph{conducts}\footnote{This word is a synonym of \enquote{behavior},
  meant to suggest that a type is a set of programs behaving in the same
  way.}. The idea is that $X$ is a set of tests, and the conduct $\cond{A} = X^\bot$
is the set of programs passing those tests. The condition $\exists X.\;\cond{A}
= X$ is equivalent to $\cond{A = A^{\bot\bot}}$; thus, $(-)^\perp$ can be used
as an involutive linear negation.

\begin{definition}
A \emph{conduct} $\cond{A}$  on a set $V(\cond{A})$ is a set of graphs with 
vertex set $V(\cond{A})$ such that $\cond{A}=\cond{A}^{\bot\bot}$.
\end{definition}


The constructions on graphs then lift to constructions on conducts.

\begin{definition}
Let $\cond{A}, \cond{B}$ be conducts with $V(\cond{A})\cap V(\cond{B})
=\emptyset$. We define $\cond{A\otimes B}=\{G \otimes H\mid G\in \cond{A}, 
H\in\cond{B}\}^{\bot\bot}$
\end{definition}

It is then natural to define the morphisms from $\cond{A}$ to $\cond{B}$ as 
the graphs belonging to the conduct $\cond{A \multimap B} = (\cond{A} \otimes 
\cond{B}^\bot)^\bot$. Unfolding this definition, one gets the realizability-style 
characterization of $\cond{A \multimap B}$ as the set of graphs mapping
elements of $\cond{A}$ to elements of $\cond{B}$.

\begin{proposition}
Let $\cond{A}, \cond{B}$ be conducts such that $V(\cond{A})\cap V(\cond{B})
=\emptyset$. 
\[\cond{A\multimap B}= \{F \mid \forall G \in \cond{A},\, F 
\perp G \otimes \emptyset_{V(\cond{B})} \text{ and } F \execution G \in 
\cond{B}\}\]
\end{proposition}

More precisely, the programs in $\cond{A \multimap B}$ are those which send 
all inputs in $\cond{A}$ to outputs in $\cond{B}$, with the side condition $F \perp 
G \otimes \emptyset_{V(\cond{B})}$ meaning that there must be no alternating 
cycle between the program and the possible inputs. This side condition also 
appears in the so-called adjunction between the tensor product and the execution:
\begin{proposition}[Adjunction]
  Let $F$, $G$ and $H$ be three graphs. Suppose that $V(F) = V(G) \sqcup V(H)$.
  Then $F \perp G \otimes H$ if and only if $F \perp G \otimes \emptyset_{V(H)}$
  and $F \execution G \perp H$.
\end{proposition}

Together, the above two facts lead to the monoidal closure of the category of
conducts and graphs. Most of the ingredients for a model of MLL are now present;
the construction will be detailed when we carry it out in the setting of
coherent graphs, which is the topic of the next section.

\begin{remark}
  A way to get rid of the side condition is to keep track of the cycles
  appearing during the execution and incorporating them in our notion of
  generalized program -- which would then consist of a graph together with a
  \enquote{set of cycles}, called a \emph{project}. The orthogonality of two
  projects $\project{a \perp b}$ can then be formulated as the fact that the set
  of cycles of the project $\project{a \execution b}$ is empty -- this is a
  \emph{focussed orthogonality} in the sense of Hyland and
  Schalk~\cite{hyland2003}, and our construction can then be seen as an instance
  of \emph{double glueing} carried out in a compact closed category.
\end{remark}

The project-based construction of a model of MLL is much more general, since
it allows one to define orthogonality as any reasonable predicate on the set
of cycles, not just emptiness. This leads to a family of models parameterized
by the choice of orthogonality. The interested reader can find details of this 
more general construction in earlier papers by the second author \cite{Seiller2012,
Seiller2016}. For reasons of simplicity, and because the selected results we will
detail in the last section do not require this more general setting, we will not use 
projects in this paper and restrict to the specific case described above. It should 
however be noted that the extension from graphs to coherent graphs explained 
in the next section could very well be performed in the more general setting of 
projects.

\section{Enriching interaction graphs with coherence}

We are now in a position to define the main object of study of this paper, by
endowing the edges of interaction graphs with a coherence space structure. The
definition of execution and orthogonality will be extended to these objects, and
this will suffice to build a model of MLL from them.

\subsection{Motivation: sparse non-deterministic proofs}

Let's say we want to interpret non-deterministic superpositions of proofs. More
formally, this corresponds to proofs in a sequent calculus enriched with the
\textsc{Sum} rule

\[ \frac{\vdash \Gamma \quad \ldots \quad \vdash \Gamma}{\vdash \Gamma} \]

An obvious solution would be to represent the set of possibilities as a formal
sum of proofs -- as is done for instance in differential linear
logic~\cite{ehrhard2006,ehrhard2017}. Indeed, in~\cite{maurel2003}, which
introduces this rule, proof equivalence is extended with associativity and
distributivity over arbitrary contexts of
\textsc{Sum}~\cite[Fig.~6]{maurel2003}. These rules can turn any proof using
\textsc{Sum} into an equivalent proof with a single \textsc{Sum} rule at the
root.

However, this may lead to an exponential blow-up of the size of the proof.
(Maurel's motivation for studying the \textsc{Sum} rule in~\cite{maurel2003} is
to characterize the complexity class \textsf{NP}; for this purpose, such a size
increase must be avoided.) We would like to have concise interpretations of
proofs in which small sub-proofs involve non-deterministic choice, without
having to duplicate the context.

The same issues arise in the theory of proof nets for Multiplicative-Additive
Linear Logic: like the \textsc{Sum} rule, the $\&$-introduction rule involves a
superposition of two proofs. Hughes and van Glabbeek's slice
nets~\cite{hughes2005} use the \enquote{formal sum} approach, and therefore, a
sequent calculus proof may have an exponentially bigger translation. In other
systems of MALL proof nets, such as monomial nets~\cite{girard1996} and conflict
nets~\cite{hughes2016}, this is avoided by morally sharing parts of the net
between different terms of the formal sum, by means of some kind of
\emph{coherence relation}. (For the former, the appendix~\cite[A.1]{girard1996}
provides an alternative presentation, equivalent to monomial weights, using a
coherence space; in the latter, the conflict relation corresponds to what we call
incoherence.)

Taking inspiration from this, we will extend interaction graphs to equip their
sets of edges with a coherence space structure. This will allow us to interpret
the \textsc{Sum} rule as a \enquote{incoherent union} of graphs in an
interaction graph model. However, it will not be sufficient to get a model of
MALL, because of obstructions\footnote{The natural encoding of additives would
  be:
\begin{itemize}[nolistsep,noitemsep]
\item for $A \oplus B$, $\mathsf{inl}(G) = G \otimes \emptyset_{B}$ and
  $\mathsf{inr}(H) = \emptyset_{A} \otimes H$,
\item for $A \with B$, $G \with H = (G \otimes \emptyset_{B}) + (\emptyset_{A}
  \otimes H)$: a proof of $A \with B$ consists of a superposition of a proof of
  $A$ and a proof of $B$.
\end{itemize}
But consider two linear maps $f : A \multimap C$ and $g : B \multimap C$, their
copairing $\langle f,g \rangle : A \oplus B \multimap C$, and some $a : A$.
Then, in general, the interpretation of $f(a)$ and that of $\langle f,g
\rangle(\mathsf{inl}(a))$ differ: the latter contains any edge of $g$ whose
endpoints are both in $C$.

This also occurs with the representation of additives
by means of \emph{sliced} interaction graphs; see the discussion
in~\cite[\S5.2]{Seiller2016} for more details and a concrete example, which can
be translated to coherent graphs.} which are common to GoI models.

Nonetheless, the consideration of this coherence space structure provides
an interesting construction which could be used to built GoI models of 
extensions of MLL+\textsc{Sum}, such as differential linear logic. Moreover,
the introduction of coherence allows for restrictions on the execution of 
two graphs that ensures finiteness of the result. Indeed in the general 
construction on graphs considered above, the execution of two finite graphs 
$F, G$ may not be finite. However, the execution of two finite coherent 
graphs will be a finite coherent graph.

\subsection{Coherent graphs: basic properties}

\begin{definition}
  A \emph{coherence relation} is a symmetric reflexive binary relation. A
  \emph{coherence space} $X$ is a set $|X|$ (the \emph{web} of $X$) equipped with
  a coherence relation written $\coh_X$ (or $\coh$ when $X$ is implied by the
  context). We use the following notations: $x \frown y :\Leftrightarrow x \coh y \wedge 
  x \neq y$, and $x \smile y :\Leftrightarrow x \not\coh y \wedge x \neq y$.

  \noindent The following binary operations are defined on coherence spaces $A, B$:
  \begin{itemize}[nolistsep, noitemsep]
  \item $\with$: $|A \with B| = |A| \sqcup |B|$, $a \frown_{A \with B} b$ for
    all $a \in A,\, b \in B$, and the rest is induced from $\coh_A$, $\coh_B$
  \item $\oplus$: $|A \oplus B| = |A| \sqcup |B|$, $a \smile_{A \oplus B} b$ for
    all $a \in A,\, b \in B$ and the rest is induced from $\coh_A$, $\coh_B$
  \end{itemize}
  \noindent Lastly, a \emph{clique} is a subset of the web whose elements are pairwise coherent.
\end{definition}

\begin{definition}
  A \emph{coherent graph} is a graph $G$ together with a coherence
  relation $\coh_{G}$ on its set of edges $E(G)$. We denote by $\cohspace{G}$
  the coherence space $(E(G),\coh_{G})$ associated with $G$.
\end{definition}

\begin{definition}
  Let $G$ and $H$ be coherent graphs. The plugging $G \plugging H$ is endowed with
  the coherence relation on edges such that $\cohspace{G\plugging
    H} \simeq \cohspace{G}\with\cohspace{H}$.

  An alternating path (or cycle) in $G \plugging H$ is \emph{coherent} if all its
  edges are pairwise coherent, i.e. if it is a clique for the associated
  coherence relation. Two coherent paths (or cycles) $\pi$, $\pi'$ are
  \emph{mutually coherent} if and only all edges in $\pi$ are coherent with all
  edges in $\pi'$.

  The \emph{execution} $G \execution H$ has as vertex set $V(G) \symmdiff V(H)$
  and contains, for each \emph{coherent alternating path} from $u$ to $v$ in $G
  \plugging H$ with $u, v \in V(G) \symmdiff V(H)$, an edge from $u$ to $v$. It
  comes with the coherent graph structure induced by mutual coherence of paths
  in $G \plugging H$.

  When $V(G) \cap V(H) = \emptyset$, we write $G \otimes H$ instead of $G
  \execution H$. This is justified by the shape of $G \otimes H$: $V(G
  \otimes H) = V(G) \sqcup V(H)$ and $\cohspace{G \otimes H} \simeq
  \cohspace{G} \with \cohspace{H}$.
\end{definition}

\begin{figure}
  \centering
  \begin{subfigure}{7.5cm}
    \centering
          \begin{tikzpicture}
        \node[vertex] (a) at (0,0) {};
        \node[vertex] (b) at (1.5,0) {};
        \node[vertex] (c) at (3,0) {};
        \node[vertex] (d) at (4.5,0) {};
\draw[ultra thick,red, ->] (a) to[bend right] (c);
\draw[ultra thick,blue, ->] (d) to[bend left] (b);

        \node[vertex] (a') at (0,1.2) {};
        \node[vertex] (b') at (1.5,1.2) {};
\draw[ultra thick, ->] (b') to[bend right] (a');
        \node[vertex] (v) at (-1.5,1.2) {};
        \draw[thick, ->] (v) to[bend left] (a');
        \draw[thick, ->] (b') to[bend right] (v);

\draw[densely dotted, ultra thick] (a) -- (a');
\draw[densely dotted, ultra thick] (b) -- (b');
      \end{tikzpicture}
    \end{subfigure}~~{\LARGE $\leadsto$}~~~\begin{subfigure}{7cm}
      \centering
      \begin{tikzpicture}
        \node[vertex] (v) at (-1.5,-2) {};
        \node[vertex] (a) at (3,-2) {};
        \node[vertex] (b) at (4.5,-2) {};

          \draw[thick,red, ->] (v) to[bend left] (a);
          \draw[thick,blue, ->] (b) to[bend right] (v);

      \end{tikzpicture}
    \end{subfigure}
    \caption{Execution of two coherent graphs}
    \label{fig:coh-ig}
\end{figure}
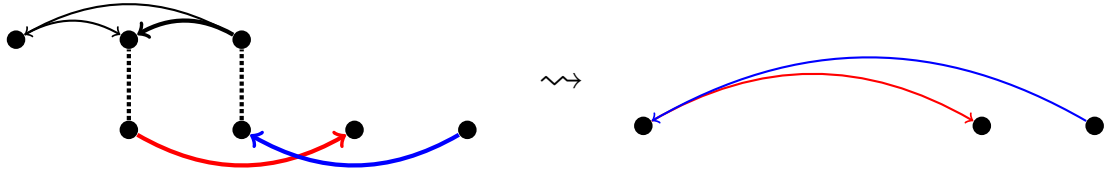

\begin{example}
  Consider the two coherent graphs on the left of Figure~\ref{fig:coh-ig}. The
  color scheme is meant to show the cliques in the coherence relation: while all
  the edges of the top graph are coherent, the blue edge and the red edge in the
  bottom graph are incoherent. Plugging them together gives a graph with:
  ${\color{red}\mathrm{red}} \coh \mathrm{black}$, ${\color{blue}\mathrm{blue}}
  \coh \mathrm{black}$, ${\color{red}\mathrm{red}} \smile
  {\color{blue}\mathrm{blue}}$.

  The execution is then computed by following the coherent alternating paths in
  this plugging, giving the two edges of the graph on the right. Note that the
  bolded path on the left is incoherent -- it contains both a red and a blue
  edge -- and therefore does not yield an edge in the execution.

  Since the two coherent paths in the plugging are mutually incoherent, the
  execution consists of two incoherent edges, which have been drawn with
  different colors to illustrate this.
\end{example}

This definition of execution is still adequate to compose morphisms:
\begin{proposition}[Associativity for coherent graphs]
  Let $F$, $G$ and $H$ be three coherent graphs. Suppose that $V(F) \cap V(G)
  \cap V(H) = \emptyset$. Then $(F::G)::H \equiv F::(G::H)$.
\end{proposition}
\begin{proof}
  Both sides of the $\equiv$ sign are ways of computing the graph whose vertices
  are $V(F) \symmdiff V(G) \symmdiff V(H)$, and whose edges correspond to the
  coherent alternating paths in $F \plugging G \plugging H$, with the coherence
  relation given by mutual coherence of paths.
\end{proof}

And the operation mentioned earlier to represent non-determinism can now be
formalized as follows:
\begin{definition}
  The \emph{incoherent sum} $G \incohplus H$ of two coherent graphs $G$ and $H$
  with $V(G) = V(H) = V$ is defined as $V(G \incohplus H) = V$, $\cohspace{G
    \incohplus H} = \cohspace{G} \oplus \cohspace{H}$ (therefore $E(G \incohplus H) = E(G)
  + E(H)$).
\end{definition}

For example, the bottom left graph of Figure~\ref{fig:coh-ig} can be seen as the
incoherent sum of the red subgraph and the blue subgraph -- both subgraphs
having 4 vertices. Likewise, the graph on the right can be decomposed into an
incoherent sum.

\subsection{An application of coherence: simple paths}\label{sec:Simple}

Although we haven't yet redefined orthogonality and conducts for coherent
graphs, what we have seen until now suffices to showcase an application of the
coherence space structure on edges, initially considered in Seiller's PhD
thesis~\cite{seillerthese}.

The motivation for the extension of graphs with a notion of coherence was to
limit the execution, which often led to infinite graphs. Indeed, let $F$ and $G$
be two graphs, such that $F \plugging G$ contains an alternating
path of the form $\pi C \pi'$, where $C$ is an alternating cycle and $\pi, \pi'$
are two alternating paths, with endpoints in $V(F) \symmdiff V(G)$. Then the
execution $F \execution G$ contains infinitely many edges, even when $F$ and $G$
are finite, since all the $\pi C^n \pi'$, $n \in \mathbb{N}$, are alternating
paths.

In the model defined in section~\ref{sec:ig}, this doesn't happen when composing
morphisms because the \enquote{side condition} ensures that $F \plugging G$ does
not contain alternating cycles. However, when looking at more general models
\cite{Seiller2012,Seiller2016}, we need another way to keep the execution
finite. It is then natural to consider restricting our attention to simple paths
and cycles.

\begin{definition}
  A path (resp.\ cycle) is \emph{simple} if it does not visit any vertex more
  than once.
\end{definition}

But restricting the execution (resp.\ orthogonality) to simple paths (resp.\
cycles) breaks the adjunction.

\begin{example}
  Figure~\ref{simplecontreex} shows how the adjunction fails when one considers
  only simple paths and cycles. Let us write $G_1 \execution_s G_2$ for the
  \emph{simple execution} between $G_1$ and $G_2$ -- i.e. the execution
  restricted to the edges coming from simple paths -- and $G_1 \perp_s G_2$ when
  there is no alternating simple cycle in $G_1 \plugging G_2$.

  On the left, we have no alternating simple cycle in $F \plugging G \plugging H$,
  so $F \perp_s G \otimes H$. On the other hand, one can see on the right that
  there is a simple alternating cycle $C$ between the simple execution $F
  \execution_s G$ and $H$, so $F \execution_s G \not\perp_s H$. In fact, this
  cycle lifts to a non-simple cycle in $F \plugging G \plugging H$: the edges
  $deb$ and $ce^{-1}a$ correspond to paths in $F \plugging G$ which both visit the
  vertices 1 and 2.
  
  This example therefore explains why the introduction of coherence is necessary
  to restrict to simple paths: coherence keeps track of possible dependencies
  between paths (i.e. two paths that used the same edge) that may be forgetten
  and lead to troubles in the model (i.e. break associativity, which is the model's
  counterpart of the Church-Rosser property).
\end{example}

\begin{figure}
\centering
\begin{subfigure}{9cm}
\begin{tikzpicture}[scale=0.7]
	\node (1) at (0,0) {\scriptsize{$1$}};
	\node (2) at (2,0) {\scriptsize{$2$}};
	\node (3) at (4,0) {\scriptsize{$3$}};
	\node (4) at (6,0) {\scriptsize{$4$}};
	\node (5) at (8,0) {\scriptsize{$5$}};
	\node (6) at (10,0) {\scriptsize{$6$}};
	
	\definecolor{darkgreen}{rgb}{0.15,0.6,0.15};
	
	\draw[->,red] (1) .. controls (4,3) and (6,3) .. (5) node [midway,above] {$a$};
	\draw[->,red] (2) .. controls (2.8,1) and (3.2,1) .. (3) node [midway,below] {$b$};
	\draw[->,red] (4) .. controls (5,2) and (3,2) .. (2) node [midway,below] {$c$};
	\draw[->,red] (6) .. controls (6,4) and (4,4) .. (1) node [midway,above] {$d$};
	
	\draw[->,blue] (1) .. controls (0.5,-0.7) and (1.5,-0.7) .. (2) node [midway,below] {$e$};
	\draw[->,blue] (2) .. controls (1.5,-1.3) and (0.5,-1.3) .. (1) node [midway,below] {$e^{-1}$};
	
	\draw[->,darkgreen] (3) .. controls (4.5,-1) and (5.5,-1) .. (4) node [midway,above] {$f$};
	\draw[->,darkgreen] (5) .. controls (8.5,-1) and (9.5,-1) .. (6) node [midway,above] {$g$};
	
	\draw[-,dashed,red] (0,0.2) -- (10,0.2) {};
	\draw[-,dashed,red] (10,4) -- (10,0.2) {};
	\draw[-,dashed,red] (0,0.2) -- (0,4) {};
	\draw[-,dashed,red] (0,4) -- (10,4) {};
	
	\draw[-,dashed,blue] (0,-0.2) -- (2,-0.2) {};
	\draw[-,dashed,blue] (2,-2) -- (2,-0.2) {};
	\draw[-,dashed,blue] (0,-0.2) -- (0,-2) {};
	\draw[-,dashed,blue] (0,-2) -- (2,-2) {};

	\draw[-,dashed,darkgreen] (4,-0.2) -- (10,-0.2) {};
	\draw[-,dashed,darkgreen] (10,-2) -- (10,-0.2) {};
	\draw[-,dashed,darkgreen] (4,-0.2) -- (4,-2) {};
	\draw[-,dashed,darkgreen] (4,-2) -- (10,-2) {};
	
	\node (G) at (0.3,-1.7) [blue] {G};
	\node (H) at (9.7,-1.7) [darkgreen] {H};
	\node (F) at (0.3,3.7) [red] {F};
\end{tikzpicture}
\caption{$F\plugging G \plugging H \cong F \plugging (G \otimes H)$}
\end{subfigure}~~~~~\begin{subfigure}{6cm}
\begin{tikzpicture}[scale=0.7]
	\node (3) at (0,0) {\scriptsize{$3$}};
	\node (4) at (2,0) {\scriptsize{$4$}};
	\node (5) at (4,0) {\scriptsize{$5$}};
	\node (6) at (6,0) {\scriptsize{$6$}};
	
	\definecolor{darkgreen}{rgb}{0.15,0.6,0.15};
	
	\draw[->,violet] (4) .. controls (2.8,1) and (3.2,1) .. (5) node [midway,below] {$ce^{-1}a$};
	\draw[->,violet] (6) .. controls (5,2) and (1,2) .. (3) node [midway,below] {$deb$};

	\draw[->,darkgreen] (3) .. controls (0.5,-1) and (1.5,-1) .. (4) node [midway,above] {$f$};
	\draw[->,darkgreen] (5) .. controls (4.5,-1) and (5.5,-1) .. (6) node [midway,above] {$g$};
	
	\draw[-,dashed,violet] (0,0.2) -- (6,0.2) {};
	\draw[-,dashed,violet] (6,2) -- (6,0.2) {};
	\draw[-,dashed,violet] (0,0.2) -- (0,2) {};
	\draw[-,dashed,violet] (0,2) -- (6,2) {};
	
	\draw[-,dashed,darkgreen] (0,-0.2) -- (6,-0.2) {};
	\draw[-,dashed,darkgreen] (6,-2) -- (6,-0.2) {};
	\draw[-,dashed,darkgreen] (0,-0.2) -- (0,-2) {};
	\draw[-,dashed,darkgreen] (0,-2) -- (6,-2) {};	

	\node (H) at (0.3,-1.7) [darkgreen] {H};
	\node (FG) at (0.6,1.7) [violet] {F$\execution_s $G};
	
\end{tikzpicture}
\caption{$(F\execution_s G)\plugging H$\label{simplecontreexex}}
\end{subfigure}

\caption{Counter-example of the adjunction for simple paths and cycles\label{simplecontreex}}
\end{figure}
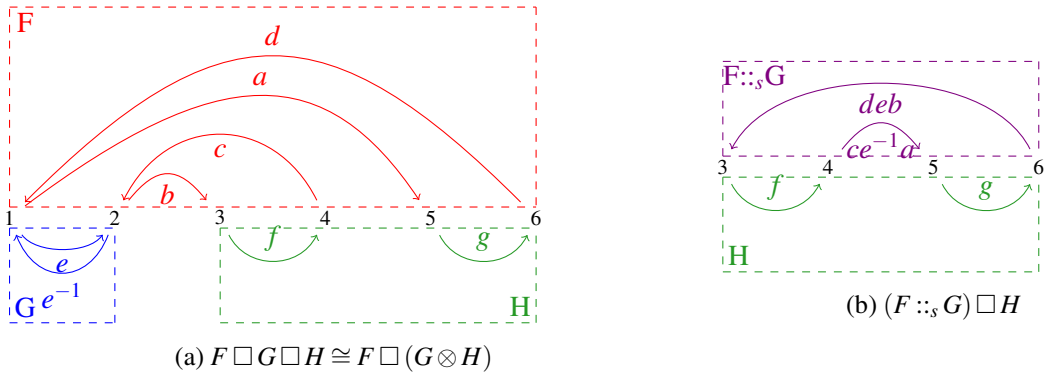

Coherent graphs were therefore a way of making the adjunction work with simple paths.

\begin{remark}
  To ensure the finiteness of execution, another natural choice would be to
  restrict to paths and cycles without repeating \emph{edges}, instead of
  vertices. This is done in a similar way to what we are about to show here,
  but using a non-reflexive coherence relation~\cite[\S 5.2]{seillerthese}.
\end{remark}

\begin{definition}
  Let $G$ be a graph. The \emph{simple coherence} relation
  on $E(G)$ is defined as follows: two edges are coherent if
  and only if they do \emph{not} have a common vertex.
\end{definition}

\begin{definition}
  A coherent graph $G$ is \emph{simple} if $\coh_{G}$
  is included in the simple coherence on $E(G)$.
\end{definition}

For instance, any coherence on a matching is simple, and the \emph{chordless
  coherence} defined in section~\ref{sec:principal} is included in the simple
coherence.

\begin{proposition}
  Let $G$ and $H$ be two simple coherent graphs with $V(G)=V(H)$. 
  A coherent alternating path or cycle in $G \plugging H$ is simple.
  The converse holds if $\coh_G$ and $\coh_H$ are equal to the simple 
  coherences.
\end{proposition}
\begin{proof}
  We prove the proposition for paths; a similar argument works for cycles.

  Let $\pi$ be an alternating path in $G \plugging H$, and suppose it is not
  simple. Among all the vertices visited at least twice, let $v$ be the one with
  the earliest first visit; this first visit involves an edge $e$ with $t(e) =
  v$. Suppose w.l.o.g.\ $e \in E(G)$. The second visit to $e$ involves an edge
  $e' \in E(G)$ (either incoming or outgoing) because of alternation; $e \neq
  e'$, or else $s(e)$ would contradict the minimality of $v$. And since $t(e) =
  v \in \{s(e'), t(e')\}$, $e \smile_G e'$ by definition of simplicity of
  $\coh_G$. Thus, $\pi$ is incoherent.

  Conversely, in an alternating simple path, the source and destination are
  incident to a single edge, while every intermediate vertex is incident to
  exactly one edge of $G$ and one edge of $H$. So the path cannot contain two
  incident edges from the same graph.
\end{proof}

Thus, as long as we execute two finite simple coherent graphs together, we
always get finite graphs. It only remains to check that simplicity is preserved
by execution to show that we have devised an universe closed under execution of
generalized programs which are all finite graphs.

\begin{proposition}
  If $G$ and $H$ are simple, then $G \execution H$ is simple.
\end{proposition}
\begin{proof}
  Let $e \neq e'$ have a common source vertex in $G \execution H$; we wish to show $e
  \smile e'$. By definition, $e$ and $e'$ correspond to paths in $G \plugging H$,
  respectively $\pi$ et $\pi'$. Let $\rho$ be their longest common prefix: $\pi
  = \rho \sigma$ and $\pi' = \rho \sigma'$. Since $\pi \neq \pi'$, at least one
  of $\sigma$ and $\sigma'$ must be non-empty. Furthermore, if only one of them
  were empty, the destination of $\pi$ would be a vertex in $V(G) \cap V(H)$
  which is impossible.
  
  Therefore, there exists $f \neq f'$ such that $\sigma = f \psi$, $\sigma' = f'
  \psi'$. Then $f \smile f'$ by simplicity of $\coh_G$ si $f, f' \in E(G)$, or
  of $\coh_H$ si $f, f' \in E(H)$ (necessarily, $f$ and $f'$ come from the same
  graph; this is where the common starting vertex assumption comes into play).
  Thus, $\pi$ and $\pi'$ are not mutually coherent.

  If $e$ and $e'$ have the same target vertex, the proof is the same in the
  opposite graph; if $s(e) = t(e')$ or $t(e) = s(e')$, it suffices to look at
  the pair of edges incident to the common vertex.
\end{proof}

\subsection{A model made out of coherent graphs}

Now, we carry out the orthogonality-based construction of a model of MLL, which
was previously sketched, in the setting of coherent graphs.

\begin{definition}
  For two coherent graphs $G, H$ with $V(G) = V(H)$, we write $G \perp H$ when $G
  \plugging H$ contains no \emph{coherent} alternating cycle.
\end{definition}

As before, the tensor and the execution satisfy an adjunction with respect to
$\perp$:

\begin{proposition}
  Let $F$, $G$ and $H$ be three coherent graphs. Suppose that $V(F) = V(G)
  \sqcup V(H)$. Then $F \perp G \otimes H$ if and only if $F \perp G \otimes
  \emptyset_{V(H)}$ and $F \execution G \perp H$.
\end{proposition}
\begin{proof}
  A coherent alternating cycle between $F$ and $G \otimes H$ either is between
  $F$ and $G$ or involves at least one edge of $H$. The former case corresponds
  to $F \perp G \otimes \emptyset_{V(H)}$. Cycles fitting in the latter case
  correspond exactly to alternating cycles between $F \execution G$ and $H$.
\end{proof}

The notion of conduct is the same as before, and actually the next few
definitions and theorems are rather formal and are not affected by adding a
coherence relation to graphs.

\begin{definition}
  A \emph{conduct} $\cond{A}$ on a vertex set $V(\cond{A})$ is a set of coherent 
  graphs with vertex set $V(\cond{A})$ such that $\cond{A}=\cond{A}^{\bot\bot}$.
\end{definition}

\begin{definition}
  Let $\cond{A}$ and $\cond{B}$ be two conducts such that $V(\cond{A}) \cap
  V(\cond{B}) = \emptyset$.
  We define $\cond{A \otimes B} = \{ G \otimes H \mid G \in
  \cond{A}, H \in \cond{B} \}^{\bot\bot}$.
  \end{definition}

It follows that  $\cond{A \parr B} = \{ G \otimes H \mid G \in \cond{A}, H \in
 \cond{B} \}^\bot$, and $\cond{A \multimap B} = \cond{A^\bot \parr B}$.
 As in the previous section, one easily obtains the following characterisation 
 of the linear maps.
 
 \begin{proposition}
 Let $\cond{A}$ and $\cond{B}$ be two conducts such that $V(\cond{A}) \cap
  V(\cond{B}) = \emptyset$. Then
  \[\cond{A\multimap B}= \{F \mid \forall G \in \cond{A},\, F 
\perp G \otimes \emptyset_{V(\cond{B})} \text{ and } F \execution G \in 
\cond{B}\}\]
 \end{proposition}

From these, one can define a categorical model of MLL, establishing that
the interpretation of MLL proofs by coherent graphs is sound. This can be 
further refined by showing that cut elimination is soundly represented by the
execution of two graphs.

Details of the proofs of these results are ghastly, as one needs to consider
\emph{delocations}. Indeed, as each conduct has a location, i.e. a specified
set of vertices, one needs to consider isomorphic copies of conducts when
defining all the basic operations. Morally, the categorical model consists of
the category whose objets are conducts and morphisms from $\cond{A}$ to
$\cond{B}$ are coherent graphs in $\cond{A\multimap B}$. In details, the 
conduct $\cond{A\multimap B}$ may not be defined if by some bad luck we
ended up with a situation where $V(\cond{A})=V(\cond{B})$. This is 
circumvented by defining the set of morphisms from $\cond{A}$ to
$\cond{B}$ are coherent graphs in $\cond{\phi(A)\multimap \psi(B)}$, where 
$\phi$ and $\psi$ rename the sets of vertices to ensure disjointness. This
leads to complications in the definition of composition. However, the 
details of the involved combinatorics is provided in full in the second author's
first paper on interaction graphs \cite{Seiller2012} and the addition of a 
coherence relation on edges of graphs do not change the proof: understanding
the word "graph" as meaning "coherent graph", the original proof can be used
word to word to obtain the following result, as only the associativity and 
adjunction properties are used.

Similarly, as the tensor product is not a total operation, it needs to be dealt 
with through delocations. Again, the method is exactly that of the previous 
work mentioned above.


\begin{theorem}
  \label{thm:star-autonomous}
  The following defines a $*$-autonomous category: conducts as objects, 
  coherent graphs as morphisms -- modulo delocations and composed by 
  execution --, delocated tensor as monoidal product, and $(-)^\bot$ as 
  dualization.
\end{theorem}

Now, recall that we motivated the introduction of the coherence relation by
mentioning a \textsc{Sum} rule, which is supposed to be interpreted by the
incoherent sum $\incohplus$. But the category of coherent conducts and graphs
does \emph{not} interpret this rule: the commutations which distributes
\textsc{Sum} over arbitrary contexts do not correspond to equalities of coherent
graphs. However, one can get a model of MLL+\textsc{Sum} by identifying two
coherent graphs when they have the same \enquote{deterministic sub-graphs},
that is, the same maximal cliques. Note this identifies coherent graphs which
are indistinguishable w.r.t.\ orthogonality. I.e. this equivalence relation is finer
than what Seiller calls \emph{universal equivalence} \cite{seillerthese} and 
which is defined as $G\sim H$ if for all graph $F$ the set of alternating cycles
in $G\plugging F$ is equal to the set of alternating cycles in $H\plugging F$. 
The universal equivalence is quite natural, and implies that the equivalent 
graphs $F,G$ are observationally indistinguishable in \emph{all} Interaction
Graphs models definable. Thus, the quotient needed to represent faithfully
the \textsc{Sum} rule is quite natural to consider and provides a very 
satisfactory model, as does the quotient w.r.t. universal equivalence.

\begin{theorem}
  Let $\mathcal{R}$ be the following equivalence relation on graphs:
  $G\,\mathcal{R}\,H$ iff $V(G) = V(H)$ and there is a bijection between the
  maximal cliques (for inclusion) of $E(G)$ and those of $E(H)$, commuting with
  the source and target maps. (That is, $G$ and $H$ \enquote{have the same}
  maximal cliques.)

  $\mathcal{R}$ is a congruence on the hom-sets of the category in the previous
  theorem, and the quotient is a model of \emph{MLL+\textsc{Sum}}, interpreting
  \emph{\textsc{Sum}} as $\incohplus$.
\end{theorem}

Of course, each equivalence class then has a canonical representative which is
just the incoherent sum of the maximal cliques. However, the use of a coherence
relation remains interesting since it potentially allows for much more concise
representatives -- recall that \emph{sparsity} of non-deterministic programs was
an initial motivation.

\section{Principal conducts and cographic proof nets}
\label{sec:principal}

\subsection{Principal conducts}

Our goal is now to study the interpretation of MLL types in the model defined by
Theorem~\ref{thm:star-autonomous} -- that is, without the quotient -- and
recover a notion of proof net together with an associated correctness criterion.
To do so we will look at the generation of conducts by single \enquote{tests} in
their orthogonal.

\begin{definition}
  Let $\cond{A}$ be a conduct. A \emph{generator} of $A$ is a graph $G \in
  \cond{A}^\perp$ such that $\cond{A} = \{G\}^\perp$.
\end{definition}

Let $\cond{A}$ be generated by a set $X$, that is, $\cond{A} = X^\perp$;
then it admits a generator, namely, the incoherent sum of all graphs in $X$
(even if $X$ is infinite, $\incohplus$ makes sense as an infinitary operation).
But this doesn't simplify much the description of $\cond{A}$. What we would like
is to find generators with a reasonable size, e.g.\ with a number of edges
polynomial in the number of vertices. Coherent interaction graphs allow us to
build a model of MLL consisting entirely of conducts with such generators. 

\begin{definition}
  A conduct $\cond{A}$ is \emph{principal} if it admits a generator without
  parallel edges, i.e.\ multiple edges with both the same sources and the same
  targets. It is \emph{bi-principal} if both $\cond{A}$ and $\cond{A}^\bot$ are
  principal.
\end{definition}

\begin{theorem}
  Principal conducts are closed under $\otimes$ and $\parr$, and bi-principal
  conducts constitute a non-trivial model of MLL.

  Thus, there is a model of MLL whose conducts all admit generators whose number
  of edges is at most \emph{quadratic} in the number of vertices (with no
  parallel edges allowed, it is bounded by the number of ordered pairs of
  vertices).
\end{theorem}

\begin{remark}
  Moreover, given a fixed assignment of atoms to bi-principal conducts, the
  number of vertices -- and thus the number of edges -- of a generator is
  polynomial in the size of the interpreted formula.
\end{remark}

The non-triviality comes from the fact that the unique conduct on $\{*\}$ is
bi-principal. This conduct does not contain any coherent graph, but this tells
us that to get a truly non-empty bi-principal conduct, we can just take any
provable formula and interpret every atom by $\{*\}$. The rest of the theorem is
a consequence of the following explicit constructions on generators.

\begin{proposition}
  For all coherent graphs $G$ and $H$ with $V(G) \cap V(H) = \emptyset$, we 
  have $\{G\}^\perp\parr \{H\}^\perp = \{G \otimes H\}^\perp$ and $\{G\}^\perp 
  \otimes\{H\}^\perp = \{G \parr H\}^\perp$, where $V(G \parr H) = V(G) \sqcup V(H)$, 
  and \[ \cohspace{G \parr H} = \cohspace{G} \oplus \cohspace{H} \oplus (V(G) \times
    V(H) \sqcup V(G) \times V(H), \emptyset) \]
  That is, for each $u \in V(G)$ and $v \in V(H)$, we add both $(u,v)$ and $(v,u)$ and
  make them incoherent with all other edges.
\end{proposition}

\begin{proof}
  First, note that $\{G\}^\perp \parr \{H\}^\perp = \{X \otimes Y \mid X \in
  \{G\}^{\perp\perp}, Y \in \{H\}^{\perp}\}^\perp$ by definition. By instantiating $X
  = G$, $Y = H$ and using the contravariance of $(-)^\perp$, we get $\{G\}^\perp
  \parr \{H\}^\perp \subseteq \{G \otimes H\}^\perp$.
  
  Conversely, let $X \in \{G\}^{\perp\perp}$, $Y \in \{H\}^{\perp\perp}$ and $Z
  \perp G \otimes H$; our goal is to show $Z \perp X \otimes Y$. To do so, we
  use the adjunction: $Z \perp G \otimes H \implies Z \execution G \perp H
  \implies Z \execution G \in \{H\}^\perp \implies Z \execution G \perp Y$
  (because $Y \in \{H\}^{\perp\perp}$) $\implies Z \perp G \otimes Y$ etc.\ and
  in the end, $Z \perp X \otimes Y$.

  As for $\{G \parr H\}^\perp$, the edges between $V(G)$ and $V(H)$ ensure that
  all graphs orthogonal to $G \parr H$ can be written as $X \otimes Y$ with
  $V(X) = V(G)$ and $V(Y) = V(H)$. $G \parr H \perp X \otimes Y$ entails $X
  \perp G$ and $Y \perp H$ (by looking at the coherent alternating cycles with
  vertices in $V(G)$, or in $V(H)$). Thus, we have $\{G \parr H\}^\perp
  \subseteq \{G\}^\perp \otimes \{H\}^\perp$. The reverse direction is proven by
  applying $(-)^{\bot\bot}$ to the inclusion $\{X \otimes Y \mid X \in
  \{G\}^\perp, Y \in \{H\}^\perp \} \subseteq \{G \parr H\}^\perp$ (note that
  $S^{\bot\bot\bot} = S^\bot$ for any set $S$). Note that this proof gives an
  \enquote{internal completeness} result as a bonus: every graph in this conduct
  decomposes as a tensor.
\end{proof}

\subsection{Cographic proof nets}

Among bi-principal conducts, we have all those built from $\{*\}$ using the
multiplicative connectives. Let us take a closer look at what those look like.
The following shows that the coherence relation of the canonical generator of
such a conduct is entirely determined by its vertices and edges.

\begin{definition}
  Let $G$ be a graph. The \emph{chordless coherence} relation is defined as
  follows: $e, f \in E(G)$ are incoherent if and only if either $e$ and $f$ are
  incident, or some $g \in E(G)$ is incident to both $e$ and $f$.
\end{definition}

\begin{proposition}
  Any generator obtained from the one-vertex graph by the constructions of the
  previous proposition is equipped with the chordless coherence relation.
\end{proposition}

The proof is a straightforward induction. Our choice of naming is justified by the following.

\begin{proposition}
  Let $G$ and $H$ be two graphs equipped with the chordless coherence relation.
  An alternating cycle in $G \plugging H$ is coherent if and only it is simple and
  there is no edge of $G \plugging H$ between two non-consecutive vertices of the
  cycle (such an edge is called a \emph{chord} in graph theory). The same holds
  for alternating paths in $G \plugging H$ with endpoints in $G \symmdiff H$.
\end{proposition}

\begin{proof}
  Let $\pi$ be an alternating cycle in $G \plugging H$, with a chord $e \in \pi$,
  and let $u = s(e)$, $v = t(e)$. Suppose w.l.o.g.\ that $e \in E(G)$. Since
  $\pi$ is alternating, there exist $e_u, e_v \in E(G) \cap \pi$ such that $u
  \in \{s(e_u), t(e_u)\}$ and similarly for $v$. Since $u$ and $v$ are
  non-consecutive, $e_u \neq e_v$. Then, because $e$ is incident to both, $e_u
  \smile_G e_v$, making $\pi$ incoherent.
  Conversely, if $\pi$ is incoherent, by finding the cause of incoherence
  between two edges of $\pi$ according to the definition, one can exhibit either
  a repeated vertex or a chord.
\end{proof}

Thus, it remains only to understand the set of graphs (without coherence)
generated by $\otimes$ and $\parr$. First, one can remark that those canonical
generators are all symmetric directed graphs; it will be more convenient to
consider them as \emph{undirected} graphs. In fact, it is common to associate an
undirected graph to a classical propositional formula by interpreting $\lor$ as
a disjoint union and $\land$ as the dual operation (see
e.g.~\cite{chaiken1989}). This is exactly what we do with MLL formulae; note
that the $\parr$ on conducts is interpreted as a $\otimes$ on generators, and
thus corresponds to the classical disjunction -- which is indeed the classical
counterpart of $\parr$. (See the example of Figure~\ref{fig:cograph}.) The
graphs obtained this way are \emph{cographs}~\cite{corneil1981}, a well-known
class of undirected graphs with many different characterizations.

\begin{figure}
  \centering
    \begin{tikzpicture}
      \node[bigvertex] (a) at (0,2.3) {$A$};
      \node[bigvertex] (b) at (1,2.3) {$B$};
      \node[bigvertex] (c) at (2,2.3) {$C$};
      \node[bigvertex] (d) at (3,2.3) {$D$};
      \node[bigvertex] (x) at (0.5,1.1) {{\color{red}$\parr$}};
      \node[bigvertex] (y) at (2.5,1.1) {{\color{red}$\parr$}};
      \node[bigvertex] (z) at (1.5,0) {{\color{red}$\otimes$}};

      \draw[thick] (x) -- (z);
      \draw[thick] (y) -- (z);
      \draw[thick] (x) -- (a);
      \draw[thick] (x) -- (b);
      \draw[thick] (y) -- (c);
      \draw[thick] (y) -- (d);
    \end{tikzpicture}\qquad\qquad\qquad\begin{tikzpicture}
      \node[bigvertex] (a) at (0,0) {$A$};
      \node[bigvertex] (b) at (0,2) {$B$};
      \node[bigvertex] (c) at (2,2) {$C$};
      \node[bigvertex] (d) at (2,0) {$D$};
      \draw[thick,red] (a) -- (c);
      \draw[thick,red] (a) -- (d);
      \draw[thick,red] (b) -- (c);
      \draw[thick,red] (b) -- (d);
    \end{tikzpicture}
  \caption{The cograph associated to the MLL formula $(A \parr B) \otimes (C
    \parr D)$}
  \label{fig:cograph}
\end{figure}
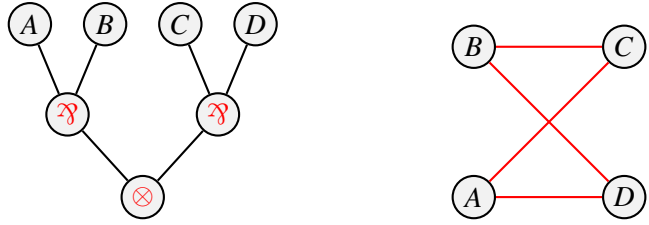

In the study of Multiplicative Linear Logic, the cograph associated to a formula
has previously been used in Retoré's version of proof nets, which he called
\enquote{R\&B-cographs}~\cite{retore2003}:
\begin{definition}
  A \emph{cographic proof} is a vertex set $V$, together with a pair of graphs
  $(G,M)$ such that $V(G) = V(M) = V$, $G$ is a cograph and $M$ is a matching;
  additionally, $V$ is labeled by MLL atoms in such a way that if two vertices
  are adjacent in $M$, then their labels are dual.
\end{definition}

This is indeed a faithful representation of proofs:
\begin{itemize}[nolistsep,noitemsep]
\item from a cograph, one can recover the original formula up to
  associativity/commutativity of $\parr/\otimes$,
\item we saw earlier that the matching associated to a proof captures its
  computational content.
\end{itemize}

But all cographic proofs do not come from proofs -- actually, they correspond to
potentially incorrect \emph{proof structures}, so this is really just another
way to describe proof nets (but one that arises naturally from our GoI
construction!). Retoré also proposed a correctness criterion for cographic
proofs; it was later rediscovered by Ehrhard~\cite{ehrhard2014} and also
underlies Hughes's \enquote{combinatorial proofs} for classical
logic~\cite{hughes2006}:

\begin{theorem}[{\cite{retore2003,ehrhard2014}}]
  Let $(G,M)$ be a cographic proof and $\Phi$ the MLL formula whose cograph is
  $G$. Then $M$ corresponds to a proof of $\Phi$ in MLL \emph{with the Mix rule}
  if and only if there is no chordless alternating cycle between $G$
  and $M$.
\end{theorem}

This \enquote{no chordless alternating cycle} condition is just orthogonality
between coherent interaction graphs: $G \perp M$ where $G$ is equipped with the
chordless coherence relation! The \enquote{only if} condition can be recovered
immediately by interpreting the corresponding proof in our model, sending the
atoms of $\Phi$ to $\{*\}$ so that $G$ is exactly the canonical generator of the
bi-principal conduct corresponding to $\Phi$, and $M$ is the element of this
conduct corresponding to the proof. (For this, our model needs to interpret the
Mix rule; this is the case since $\cond{A \otimes B} \subseteq \cond{A \parr B}$
holds for any two conducts $\cond{A,B}$.)

As for the \enquote{if} direction, it amounts to a \emph{sequentialization
  theorem} on proof nets, and thus it can't simply fall out of the above work on
coherent graphs: a proof of sequentialization must necessarily involve some hard
combinatorics. But translated into our GoI model, it almost gives a \emph{full
  completeness} theorem. Any $M$ which is orthogonal to a cograph $G$ is a proof
of the corresponding formula, up to two morphological conditions: $M$ must be a
matching, and its edges must have dual atoms as endpoints.

\section{Perspectives}

The obvious direction for future work is to extend the model to larger fragments
of linear logic. One could extend our model of MELL to handle exponentials, for
instance using \emph{graphings} such as in previous work~\cite{Seiller2017}. But
to keep finiteness and sparsity properties, it might be preferable to find a
way to use finite objects instead of measure spaces of infinite cardinality.
The ability to take non-deterministic sums in our model also suggests an
extension to differential linear logic~\cite{ehrhard2006,ehrhard2017}. It would
be interesting to understand the \enquote{correctness criterion} for
differential nets with promotion boxes arising from GoI orthogonality; we might 
recover Pagani's visible acyclicity~\cite{pagani2012} which is already 
characterised interactively. 

The cographic representation of formulae and proofs could also lead to the
discovery of new operations in a multiplicative setting. For instance, \emph{modular 
decomposition} of graphs (see e.g.~\cite{modular}) extends to arbitrary graphs the 
correspondence between cographs and formulae. This may be used to define 
generalized multiplicative connectives.

\nocite{*}
\bibliographystyle{eptcs}
\bibliography{coherent-graphs}

\begin{thebibliography}{10}
\providecommand{\bibitemdeclare}[2]{}
\providecommand{\surnamestart}{}
\providecommand{\surnameend}{}
\providecommand{\urlprefix}{Available at }
\providecommand{\url}[1]{\texttt{#1}}
\providecommand{\href}[2]{\texttt{#2}}
\providecommand{\urlalt}[2]{\href{#1}{#2}}
\providecommand{\doi}[1]{doi:\urlalt{http://dx.doi.org/#1}{#1}}
\providecommand{\bibinfo}[2]{#2}

\bibitemdeclare{article}{chaiken1989}
\bibitem{chaiken1989}
\bibinfo{author}{Seth \surnamestart Chaiken\surnameend},
  \bibinfo{author}{Neil~V. \surnamestart Murray\surnameend} \&
  \bibinfo{author}{Erik \surnamestart Rosenthal\surnameend}
  (\bibinfo{year}{1989}): \emph{\bibinfo{title}{An application of {P}4-free
  graphs in theorem-proving}}.
\newblock {\sl \bibinfo{journal}{Annals of the New York Academy of Sciences}}
  \bibinfo{volume}{555}(\bibinfo{number}{1}), pp. \bibinfo{pages}{106--121},
  \doi{10.1111/j.1749-6632.1989.tb22442.x}.

\bibitemdeclare{article}{corneil1981}
\bibitem{corneil1981}
\bibinfo{author}{D.~G. \surnamestart Corneil\surnameend},
  \bibinfo{author}{H.~\surnamestart Lerchs\surnameend} \&
  \bibinfo{author}{L.~Stewart \surnamestart Burlingham\surnameend}
  (\bibinfo{year}{1981}): \emph{\bibinfo{title}{Complement reducible graphs}}.
\newblock {\sl \bibinfo{journal}{Discrete Applied Mathematics}}
  \bibinfo{volume}{3}(\bibinfo{number}{3}), pp. \bibinfo{pages}{163--174},
  \doi{10.1016/0166-218X(81)90013-5}.

\bibitemdeclare{inproceedings}{ehrhard2014}
\bibitem{ehrhard2014}
\bibinfo{author}{Thomas \surnamestart Ehrhard\surnameend}
  (\bibinfo{year}{2014}): \emph{\bibinfo{title}{A new correctness criterion for
  {MLL} proof nets}}.
\newblock In: {\sl \bibinfo{booktitle}{Joint {Meeting} of the {Twenty}-{Third}
  {EACSL} {Annual} {Conference} on {Computer} {Science} {Logic} ({CSL}) and the
  {Twenty}-{Ninth} {Annual} {ACM}/{IEEE} {Symposium} on {Logic} in {Computer}
  {Science} ({LICS}), {CSL}-{LICS} '14, {Vienna}, {Austria}, {July} 14 - 18,
  2014}}, pp. \bibinfo{pages}{38:1--38:10}, \doi{10.1145/2603088.2603125}.

\bibitemdeclare{article}{ehrhard2017}
\bibitem{ehrhard2017}
\bibinfo{author}{Thomas \surnamestart Ehrhard\surnameend}
  (\bibinfo{year}{2017}): \emph{\bibinfo{title}{An introduction to differential
  linear logic: proof-nets, models and antiderivatives}}.
\newblock {\sl \bibinfo{journal}{Mathematical Structures in Computer Science}},
  pp. \bibinfo{pages}{1--66}, \doi{10.1017/S0960129516000372}.

\bibitemdeclare{article}{ehrhard2006}
\bibitem{ehrhard2006}
\bibinfo{author}{Thomas \surnamestart Ehrhard\surnameend} \&
  \bibinfo{author}{Laurent \surnamestart Regnier\surnameend}
  (\bibinfo{year}{2006}): \emph{\bibinfo{title}{Differential interaction
  nets}}.
\newblock {\sl \bibinfo{journal}{Theoretical Computer Science}}
  \bibinfo{volume}{364}(\bibinfo{number}{2}), pp. \bibinfo{pages}{166--195},
  \doi{10.1016/j.tcs.2006.08.003}.

\bibitemdeclare{article}{multiplicatives}
\bibitem{multiplicatives}
\bibinfo{author}{Jean-Yves \surnamestart Girard\surnameend}
  (\bibinfo{year}{1987}): \emph{\bibinfo{title}{Multiplicatives}}.
\newblock {\sl \bibinfo{journal}{Rendiconti del Seminario Matematico.
  Universitá e Politecnico di Torino}}
  \bibinfo{volume}{45}(\bibinfo{number}{Special Issue}).

\bibitemdeclare{incollection}{Girard-towards}
\bibitem{Girard-towards}
\bibinfo{author}{Jean-Yves \surnamestart Girard\surnameend}
  (\bibinfo{year}{1989}): \emph{\bibinfo{title}{Towards a Geometry of
  Interaction}}.
\newblock In \bibinfo{editor}{J.~W. \surnamestart Gray\surnameend} \&
  \bibinfo{editor}{A.~\surnamestart Scedrov\surnameend}, editors: {\sl
  \bibinfo{booktitle}{Categories in Computer Science and Logic: Proc.\ of the
  Joint Summer Research Conference}}, \bibinfo{publisher}{American Mathematical
  Society}, \bibinfo{address}{Providence, RI}, pp. \bibinfo{pages}{69--108},
  \doi{10.1090/conm/092/1003197}.

\bibitemdeclare{inproceedings}{girard1996}
\bibitem{girard1996}
\bibinfo{author}{Jean-Yves \surnamestart Girard\surnameend}
  (\bibinfo{year}{1996}): \emph{\bibinfo{title}{Proof-nets: {The} parallel
  syntax for proof-theory}}.
\newblock In: {\sl \bibinfo{booktitle}{Logic and {Algebra}}},
  \bibinfo{publisher}{Marcel Dekker}, pp. \bibinfo{pages}{97--124}.

\bibitemdeclare{article}{girard2011}
\bibitem{girard2011}
\bibinfo{author}{Jean-Yves \surnamestart Girard\surnameend}
  (\bibinfo{year}{2011}): \emph{\bibinfo{title}{Geometry of {Interaction} {V}:
  {Logic} in the hyperfinite factor}}.
\newblock {\sl \bibinfo{journal}{Theoretical Computer Science}}
  \bibinfo{volume}{412}(\bibinfo{number}{20}), pp. \bibinfo{pages}{1860--1883},
  \doi{10.1016/j.tcs.2010.12.016}.

\bibitemdeclare{article}{modular}
\bibitem{modular}
\bibinfo{author}{Michel \surnamestart Habib\surnameend} \&
  \bibinfo{author}{Christophe \surnamestart Paul\surnameend}
  (\bibinfo{year}{2010}): \emph{\bibinfo{title}{A survey of the algorithmic
  aspects of modular decomposition}}.
\newblock {\sl \bibinfo{journal}{Computer Science Review}}
  \bibinfo{volume}{4}(\bibinfo{number}{1}), pp. \bibinfo{pages}{41--59},
  \doi{10.1016/j.cosrev.2010.01.001}.

\bibitemdeclare{inproceedings}{hughes2016}
\bibitem{hughes2016}
\bibinfo{author}{Dominic \surnamestart Hughes\surnameend} \&
  \bibinfo{author}{Willem \surnamestart Heijltjes\surnameend}
  (\bibinfo{year}{2016}): \emph{\bibinfo{title}{Conflict nets:efficient locally
  canonical {MALL} proof nets}}.
\newblock In: {\sl \bibinfo{booktitle}{Proceedings of the 31st {Annual}
  {ACM}/{IEEE} {Symposium} on {Logic} in {Computer} {Science} ({LICS}), 2016}},
  \bibinfo{publisher}{ACM}, \bibinfo{address}{New York, U. S. A.}, pp.
  \bibinfo{pages}{437--446}, \doi{10.1145/2933575.2934559}.

\bibitemdeclare{article}{hughes2005}
\bibitem{hughes2005}
\bibinfo{author}{Dominic J.~D. \surnamestart Hughes\surnameend} \&
  \bibinfo{author}{Rob~J. \surnamestart Van~Glabbeek\surnameend}
  (\bibinfo{year}{2005}): \emph{\bibinfo{title}{Proof {Nets} for {Unit}-free
  {Multiplicative}-additive {Linear} {Logic}}}.
\newblock {\sl \bibinfo{journal}{ACM Transactions on Computational Logic}}
  \bibinfo{volume}{6}(\bibinfo{number}{4}), pp. \bibinfo{pages}{784--842},
  \doi{10.1145/1094622.1094629}.

\bibitemdeclare{article}{hughes2006}
\bibitem{hughes2006}
\bibinfo{author}{Dominic~J.D. \surnamestart Hughes\surnameend}
  (\bibinfo{year}{2006}): \emph{\bibinfo{title}{Proofs {Without} {Syntax}}}.
\newblock {\sl \bibinfo{journal}{Annals of Mathematics}}
  \bibinfo{volume}{143}(\bibinfo{number}{3}), pp. \bibinfo{pages}{1065--1076},
  \doi{10.4007/annals.2006.164.1065}.

\bibitemdeclare{article}{hyland2003}
\bibitem{hyland2003}
\bibinfo{author}{Martin \surnamestart Hyland\surnameend} \&
  \bibinfo{author}{Andrea \surnamestart Schalk\surnameend}
  (\bibinfo{year}{2003}): \emph{\bibinfo{title}{Glueing and orthogonality for
  models of linear logic}}.
\newblock {\sl \bibinfo{journal}{Theoretical Computer Science}}
  \bibinfo{volume}{294}(\bibinfo{number}{1}), pp. \bibinfo{pages}{183--231},
  \doi{10.1016/S0304-3975(01)00241-9}.

\bibitemdeclare{article}{kelly1980}
\bibitem{kelly1980}
\bibinfo{author}{G.~M. \surnamestart Kelly\surnameend} \&
  \bibinfo{author}{M.~L. \surnamestart Laplaza\surnameend}
  (\bibinfo{year}{1980}): \emph{\bibinfo{title}{Coherence for compact closed
  categories}}.
\newblock {\sl \bibinfo{journal}{Journal of Pure and Applied Algebra}}
  \bibinfo{volume}{19}, pp. \bibinfo{pages}{193--213},
  \doi{10.1016/0022-4049(80)90101-2}.

\bibitemdeclare{article}{semigroups}
\bibitem{semigroups}
\bibinfo{author}{Ganna \surnamestart Kudryavtseva\surnameend} \&
  \bibinfo{author}{Volodymyr \surnamestart Mazorchuk\surnameend}
  (\bibinfo{year}{2009}): \emph{\bibinfo{title}{On three approaches to
  conjugacy in semigroups}}.
\newblock {\sl \bibinfo{journal}{Semigroup Forum}}
  \bibinfo{volume}{78}(\bibinfo{number}{1}), pp. \bibinfo{pages}{14--20},
  \doi{10.1007/s00233-008-9047-7}.

\bibitemdeclare{inproceedings}{maurel2003}
\bibitem{maurel2003}
\bibinfo{author}{François \surnamestart Maurel\surnameend}
  (\bibinfo{year}{2003}): \emph{\bibinfo{title}{Nondeterministic {Light}
  {Logics} and {NP}-{Time}}}.
\newblock In: {\sl \bibinfo{booktitle}{Typed {Lambda} {Calculi} and
  {Applications}}}, \bibinfo{series}{Lecture {Notes} in {Computer} {Science}},
  \bibinfo{publisher}{Springer, Berlin, Heidelberg}, pp.
  \bibinfo{pages}{241--255}, \doi{10.1007/3-540-44904-3_17}.

\bibitemdeclare{inbook}{Seiller-Axioms}
\bibitem{Seiller-Axioms}
\bibinfo{author}{Alberto \surnamestart Naibo\surnameend},
  \bibinfo{author}{Mattia \surnamestart Petrolo\surnameend} \&
  \bibinfo{author}{Thomas \surnamestart Seiller\surnameend}
  (\bibinfo{year}{2016}): \emph{\bibinfo{title}{On the Computational Meaning of
  Axioms}}, pp. \bibinfo{pages}{141--184}.
\newblock \bibinfo{publisher}{Springer International Publishing},
  \doi{10.1007/978-3-319-26506-3_5}.

\bibitemdeclare{article}{pagani2012}
\bibitem{pagani2012}
\bibinfo{author}{Michele \surnamestart Pagani\surnameend}
  (\bibinfo{year}{2012}): \emph{\bibinfo{title}{Visible acyclic differential
  nets, {Part} {I}: {Semantics}}}.
\newblock {\sl \bibinfo{journal}{Annals of Pure and Applied Logic}}
  \bibinfo{volume}{163}(\bibinfo{number}{3}), pp. \bibinfo{pages}{238--265},
  \doi{10.1016/j.apal.2011.09.001}.

\bibitemdeclare{article}{retore2003}
\bibitem{retore2003}
\bibinfo{author}{Christian \surnamestart Retor\'e\surnameend}
  (\bibinfo{year}{2003}): \emph{\bibinfo{title}{Handsome proof-nets: perfect
  matchings and cographs}}.
\newblock {\sl \bibinfo{journal}{Theoretical Computer Science}}
  \bibinfo{volume}{294}(\bibinfo{number}{3}), pp. \bibinfo{pages}{473--488},
  \doi{10.1016/S0304-3975(01)00175-X}.

\bibitemdeclare{article}{Seiller2012}
\bibitem{Seiller2012}
\bibinfo{author}{Thomas \surnamestart Seiller\surnameend}
  (\bibinfo{year}{2012}): \emph{\bibinfo{title}{Interaction graphs:
  {Multiplicatives}}}.
\newblock {\sl \bibinfo{journal}{Annals of Pure and Applied Logic}}
  \bibinfo{volume}{163}(\bibinfo{number}{12}), pp. \bibinfo{pages}{1808--1837},
  \doi{10.1016/j.apal.2012.04.005}.

\bibitemdeclare{phdthesis}{seillerthese}
\bibitem{seillerthese}
\bibinfo{author}{Thomas \surnamestart Seiller\surnameend}
  (\bibinfo{year}{2012}): \emph{\bibinfo{title}{Logique dans le {Facteur}
  {Hyperfini}: {G\'eom\'etrie} de l'{Interaction} et {Complexit\'e}}}.
\newblock \bibinfo{type}{Th\`ese de doctorat}, \bibinfo{school}{Aix-Marseille
  Universit\'e}.
\newblock
  \urlprefix\url{https://tel.archives-ouvertes.fr/tel-00768403/document}.

\bibitemdeclare{article}{Seiller2016}
\bibitem{Seiller2016}
\bibinfo{author}{Thomas \surnamestart Seiller\surnameend}
  (\bibinfo{year}{2016}): \emph{\bibinfo{title}{Interaction graphs:
  {Additives}}}.
\newblock {\sl \bibinfo{journal}{Annals of Pure and Applied Logic}}
  \bibinfo{volume}{167}(\bibinfo{number}{2}), pp. \bibinfo{pages}{95--154},
  \doi{10.1016/j.apal.2015.10.001}.

\bibitemdeclare{article}{Seiller2017}
\bibitem{Seiller2017}
\bibinfo{author}{Thomas \surnamestart Seiller\surnameend}
  (\bibinfo{year}{2017}): \emph{\bibinfo{title}{Interaction graphs:
  {Graphings}}}.
\newblock {\sl \bibinfo{journal}{Annals of Pure and Applied Logic}}
  \bibinfo{volume}{168}(\bibinfo{number}{2}), pp. \bibinfo{pages}{278--320},
  \doi{10.1016/j.apal.2016.10.007}.

\end{thebibliography}

\end{document}